\documentclass[submission,copyright,creativecommons]{eptcs}
\providecommand{\event}{GandALF 2025} 

\usepackage{iftex}

\ifpdf
  \usepackage{underscore}         
  \usepackage[T1]{fontenc}        
\else
  \usepackage{breakurl}           
\fi

\usepackage{graphicx} 
\usepackage{amsthm}
\usepackage{amsmath}
\usepackage{amssymb}
\usepackage{proof}
\usepackage{tikz-cd}
\usepackage{enumerate}
\usepackage{cleveref}

\newcommand{\cat}[1]{\mathcal{#1}}
\newcommand{\ncat}[1]{\mathsf{#1}}
\newcommand{\EMet}{\ncat{EMet}}
\newcommand{\CMet}{\ncat{CMet}}

\newcommand{\PMet}{\ncat{PMet}}
\newcommand{\Met}{\ncat{Met}}
\newcommand{\Set}{\ncat{Set}}
\newcommand{\lt}[1]{\mathcal{#1}} 
\newcommand{\et}[1]{\mathbb{#1}} 

\newcommand{\Mod}{Mod}

\newcommand{\Mnd}{Mnd}

\newcommand{\Free}{Free}

\newcommand{\Lan}{Lan}
\newcommand{\ok}{~ok}
\newcommand{\op}{op}

\newcommand{\Law}{Law}
\newcommand{\colim}{colim}

\newcommand{\Rp}{\mathbf{R}_{\geq 0}}
\newcommand{\lT}{\lt{T}}
\newcommand{\eT}{\et{T}}
\newcommand{\e}{\varepsilon}

\newtheorem{theorem}{Theorem}[section]
\newtheorem{corollary}[theorem]{Corollary}
\newtheorem{lemma}[theorem]{Lemma}
\newtheorem{definition}[theorem]{Definition}

\newtheorem{proposition}[theorem]{Proposition}
\newtheorem{example}[theorem]{Example}

\title{Metric Equational Theories}
\author{Radu Mardare\footnote{Mardare's research was supported by the EPSRC grant EP/Y000455/1, \textit{A correct-by-construction approach to approximate computation} and by ARIA TA1.1 project \textit{Predicate Logic as a Foundation for Verified ML} }
\institute{Heriot-Watt University\\ Edinburgh, Scotland}
\and
Neil Ghani  
\institute{University of Strathclyde\\ Glasgow, Scotland}
\and
Eigil Rischel  
\institute{University of Strathclyde\\ Glasgow, Scotland}
}
\def\titlerunning{Metric Equational Theories}
\def\authorrunning{R. Mardare, N. Ghani \& E. Rischel}
\begin{document}
\maketitle

\begin{abstract}

This paper proposes appropriate sound and complete proof systems for algebraic structure over metric spaces by combining the development of Quantitative Equational Theories (QET) with the Enriched Lawvere Theories. We extend QETs to Metric Equational Theories (METs) where operations no longer have finite sets as arities (as in QETs and in the general theory of universal algebras), but arities are now drawn from countable metric spaces. This extension is inspired by the theory of Enriched Lawvere Theories which suggests that the arities of operations should be the $\lambda$-presentable objects of the underlying $\lambda$-accessible category. In this setting, the validity of terms in METs can no longer be guaranteed independently of the validity of equations, as it is the case with QET. We solve this problem, and adapt the sound and complete proof system for QETs to these more general METs, taking advantage of the specific structure of metric spaces.
\end{abstract}

\section{Introduction}


Recently, Mardare et al \cite{raduQAR2016,raduQAR2017} introduced Quantitative Algebra  (QA) and Quantitative Equational Theories (QETs) to extend one of 
the central pillars of modern mathematics, namely universal algebra (UA), from the exact world to the approximate world. Central to their work was the introduction of approximate equations $s =_\epsilon t$ (for a positive real $\epsilon$), formalising the intuition that $\epsilon$ measures the behavioural \textit{similarity} between terms $s$ and $t$. 
The generality of this new idea ---replacing Boolean reasoning within equational logic with a more refined approximate form of reasoning--- gives us a new paradigm which supports a rigorous logical framework for a proper approximation theory, where bounds can be handled, convergences proven, and limits approximated. Quantitative Equational Theories~\cite{raduQAR2016} have been used to provide simple axiomatic presentations of well-behaved metrics for several fundamental computational structures, e.g.\ the Kantorovich-Wasserstein metrics (resp.\  Hausdorff metrics) arise from convex structures (resp.\  semi-lattices). 
Further, Quantitative Algebra and Quantitative Equational Theories have 
been shown to have good meta-theoretic properties. For example, 
variety and quasi-variety results have been proven for Quantitative Algebra \cite{raduQAR2017}, revealing new insights of approximated reasoning. Similarly, compositionality principles have been studied for Quantitative Algebra, providing a formal tool to control the error propagation when computational systems interact \cite{Bacci18, Bacci21}. Further work was done in the direction of developing products and coproducts, and tensors of QETs \cite{bacci2022sum}. \textit{Conway} and \textit{iteration theories} have also been generalised to the quantitative case, to cover not only exact fixed-points $x = f (x)$, but also approximate fixed-points $x =_\epsilon f (x)$ \cite{Mardare21}. However, there remains a fundamental unanswered question:
\begin{quote}
	{\em What is the natural format/foundation for presenting algebraic structure in a metric setting? }
\end{quote}
In the setting of traditional universal algebra, there are strong results which (when adapted to the metric setting) form desiderata for any answer to the above question. These include:
(i) every equational theory supports a sound and complete proof system inducing a congruence between terms; 
(ii) every such theory induces a finitary monad on the category of $\Set$ which generates free algebras from generator sets;
(iii) the models of an equational theory are the Eilenberg-Moore algebras of the associated finitary monad; and
(iv) every finitary monad on $\Set$ arises via this construction.
 
Unfortunately, QETs fail these desiderata. For example, the monad mapping a metric space to its Cauchy completion does not arise from a QET ---we discuss this example in \cref{sec:CC}. However, the above desiderata may be framed in terms of enriched Lawvere Theories, following~\cite{power1999}, while 
turning from the case of sets with algebraic structures to metric spaces with algebraic structures. 
But there is a fundamental problem in doing so: the category $\Met$ of extended metric spaces is not locally finitely presentable but countably presentable. This means that the arities of operators and equations will not be finite sets (as for traditional algebraic theories), nor countable (discrete) metric spaces, and we will have to consider countably-presentable monads in the desiderata above. The good news is that, once these changes have been made, the work of~\cite{Kelly1993,power1999} seems to give us exactly what we need. That is i) a notion of algebraic theory consisting of operations and equations; ii) free algebras for such algebraic theories correspond to countably presentable {\bf enriched} monads on $\Met$ such that the models of the former are the Eilenberg-Moore algebras of the latter; and iii) a notion of Lawvere theory giving a syntax-free presentation of algebraic theories. The bad news is that 
%
the presentations one gets from this framework, while theoretically elegant, are very cumbersome to the point of being of little practical use. In particular, given a countably presentable monad, the associated algebraic theory has as operators of a given arity all elements of the monad applied to that arity. In the case of the Cauchy completion, where we would like one operator, we get a countably infinite number of operators. A further limitation of the enriched Lawvere theories is that they don't cover sound and complete proof systems. 

This paper, therefore, starts from the premise of wanting the best of both worlds, i.e., the theoretical clarity of Lawvere theories and the concise presentations and proof systems of QETs. Thus, we introduce {\em Metric Equational Theories} which extend QETs by allowing operators to have countable metric spaces as arities. This allows us to cover, for example, the countable presentable Cauchy Completion monad. This generalisation, however, creates a difficulty. In both traditional equation theories and QETs, one can first define the terms and then use the terms in equations. This is no longer possible with METs, as metric arities require the use of equations in defining the terms, as some terms might only be defined given that their subterms satisfy certain equations.
However, once we devise mechanisms for handling this increased complexity in METs, we will indeed get all of the above desiderata when suitably generalised to $\Met$, together with sound and complete proof systems for deriving the equality between terms.

\textbf{Outline of the paper.}
In \cref{sec:ELT} we present aspects of the category theory of metric spaces, and the theory of enriched Lawvere theories as it specializes to this case.
In \cref{sec:MET} we introduce the Metric Equational Theories, their categories of models, and proving basic properties of these. In \cref{sec:monad}, we construct a free-forgetful adjunction for each MET, and prove that these correspond to enriched Lawvere theories. 
In \cref{sec:LawisMET}, we carry out the other direction, constructing for each enriched Lawvere theory $\lT$ over $\Met$ an MET whose category of models is equivalent to $\Mod(\lT)$. 
Finally, we prove a completeness theorem for METs in \cref{sec:Completeness}, discuss the MET of Cauchy completeness in \cref{sec:Cauchy}, and consider some special classes of METs and study their corresponding monads in \cref{sec:Specials} ---in particular, we characterize the class of monads axiomatized by ordinary QETs. The paper includes an Appendix containing details proof os some of the results presented.

\section{Enriched Lawvere theories in $\Met$}
\label{sec:ELT}

Our focus is on the category $\Met$ where the objects are \emph{extended metric spaces},  i.e., sets $X$ equipped with a metric $d: X \times X \to [0,\infty]$; and the morphisms are the \emph{nonexpansive} maps, \cite{ros2020metric}.
The \textit{tensor product} $\otimes$ on $\Met$ is given by $(X,d_X) \otimes (Y,d_Y) = (X \times Y, d_{X \otimes Y})$,
where $d_{X \otimes Y}((x,y),(x',y')) = d_X(x,x') + d_Y(y,y')$.

$\Met$ is a \textit{symmetric monoidal category}, with associator, unitor and symmetry given as for the Cartesian symmetric monoidal structure on $\Set$.
Moreover, $\Met$ is closed as a monoidal category, with the internal hom $[X,Y]$ being given by the set of nonexpansive maps $X \to Y$ in the metric $d_{[X,Y]}(f,g) = \sup_x d_Y(f(x),g(x))$.
Moreover, $\Met$ is \textit{countably locally presentable} \cite{Lieberman2017}, and its countable objects are precisely those metric spaces with a countable underlying set \cite{Adamek2022}.

The concept of \textit{enriched Lawvere theory} was proposed by Power in \cite{power1999}.
Up to equivalence, a Lawvere theory is a category with finite products, whose objects are generated by a distinguished object under finite products.
These then classify finitary monads on $\Set$. Power states these results for the finitary case, remarking although that they generalize straightforwardly to any regular cardinal.

Given a regular cardinal $\kappa$ and a $\kappa$-presentable biclosed monoidal category $\cat{V}$,
we obtain a theory of $\cat{V}$-enriched Lawvere theories.
These are defined to be $\cat{V}$-categories generated under products of cardinality less than $\kappa$ and powers by $\kappa$-small objects of $\cat{V}$ by a distinguished object.
These then classify strong (i.e enriched) $\kappa$-presentable monads on $\cat{V}$.
The classical Lawvere theory is obtained by taking $\cat{V} = \Set$ and $\kappa = \aleph_0$.
In this paper, we will consider the case in which $\cat{V} = \Met$ and $\kappa = \aleph_1$. Since $\Met$ is not finitely presentable, we will be looking at \emph{countable-arity} operations in general, but we will also discuss the subset of finitary monads.
To simplify the language, in what follows we will refer to $\aleph_1$-locally presentable categories as \emph{countably locally presentable}, $\aleph_1$-accessible monads as \emph{countable-arity monads}, and so on.


Recall that in a category $\cat{C}$ enriched over $\cat{V}$, a \emph{power} or \emph{cotensor} of an object $X \in \cat{C}$ by $V \in \cat{V}$ is an object $X^V$ so that $\cat{C}(A, X^V) \cong [V,\cat{C}(A,X)]$ (with $[-,-]$ being the internal hom in $\cat{V}$).
In the self-enrichment of $\cat{V}$, powers are given by the internal hom $W^V = [V,W]$.
In the present case of metric spaces, we generally prefer the exponential notation $Y^X$ for this space, which is given by the set of nonexpansive maps $X \to Y$ in the sup-metric $d(f,g) = \sup_x d_Y(f(x),g(x))$.

Note that a category enriched in $\Met$ is simply an ordinary category $\cat{C}$ equipped with a metric on each hom-set,
so that the composition operation $\circ: \cat{C}(Y,Z) \otimes \cat{C}(X,Y) \to \cat{C}(X,Z)$ is nonexpansive.
A functor $\cat{C} \to \cat{D}$ in the enriched sense is simply a functor $F$ between the underlying categories so that each map $F: \cat{C}(X,Y) \to \cat{D}(FX,FY)$ is nonexpansive.
In particular, being enriched is a \emph{property} of an ordinary functor, not extra data.

We will primarily be interested in monads on $\Met$ itself --- here an ``enriched monad'' is just an ordinary monad $(T,\mu,\eta)$ on the category $\Met$ where the functor $T$ is nonexpansive, and the Eilenberg-Moore category of such a monad is just the EM category in the ordinary sense, metrized with the $\sup$-metric restricted to the subset of homomorphisms.

An enriched functor is also called \emph{strong} (because for a closed monoidal category $\cat{V}$, an enriched endofunctor on $\cat{V}$ is equivalent to a strong endofunctor, \cite[Theorem 1.3]{Kock1972}), and we will also use the term \emph{strong monad}.
Almost every functor we work with will be enriched in this sense, so we will often not bother to be precise about the difference.

Recall that a \emph{pseudometric space} is the generalization of metric spaces without the requirement that $d(x,y)=0 \Rightarrow x=y$. There is an obvious category $\PMet$ of pseudometric spaces and nonexpansive maps.
$\Met$ is a full subcategory of $\PMet$, and this inclusion is reflective --- for each pseudometric space $X$, the reflection is given by the quotient $X/\sim$ where $x \sim y$ if $d(x,y) = 0$.
We refer to this as the \emph{metric quotient} of $X$. It will often be convenient to describe metric spaces in this way.

	Let $\CMet$ denote the full subcategory of $\Met$ of countable extended metric spaces.
	$\CMet$ inherits a $\Met$ enrichment.
	Since $\CMet$ is closed under tensor products, the opposite category $\CMet^{\op}$, has all powers by countable objects (and they are simply given by the tensor product).

\begin{definition}
	A \emph{$\Met$ Lawvere theory} is a $\Met$-category $\cat{L}$ equipped with an identity-on-objects
	functor $\CMet^{\op} \to \cat{L}$ which is enriched and preserves powers by countable objects.
\end{definition}

	Note that every object $X$ in $\CMet^{\op}$ can be written as a (countable) power of the singleton $*$ by $X$.
	Hence, if $F$ is the functor $\CMet^{\op} \to \cat{L}$, 
	we have $F(X) = F(*)^X$, and these are all the objects of $\cat{L}$, up to isomorphism.
	Note also that $\CMet^{\op}$ has products (given by coproducts in $\Met$), and these are automatically preserved by $F$, since $$F(X \sqcup Y) = F(*)^{X \sqcup Y} = F(*)^X \times F(*)^Y.$$
	Hereafter, we abuse the notation and denote a $\Met$ Lawvere theory $(\cat{L}, F: \CMet^{\op} \to \cat{L})$ simply by $\cat{L}$.
	We denote the object $F(*)$ by $x$, and use $x^X$, with $X$ a countable metric space, to denote a general object of $\cat{L}$.
Note that, by definition of powers, the hom-object $\cat{L}(x^A,x^B)$ is isomorphic to the metric space $[B,\cat{L}(x^A,x)]$ of maps $B \to \cat{L}(x^A,x)$.
Hence, the hom-objects of a $\Met$ Lawvere theory are determined by the spaces $\cat{L}(x^A,x)$.
We call	the metric space $\cat{L}(x^A,x)$, the \emph{space of $A$-ary operations} of $\cat{L}$.
\begin{definition}
	A \emph{model} of a $\Met$ Lawvere theory $\cat{L}$ is a power-preserving $\Met$-functor
	$M: \cat{L} \to \Met$.
	The \emph{category} $\Mod(\cat{L})$ of models is the category of such functors and their natural transformations.
\end{definition}
Note that $\Mod(\cat{L})$ carries a $\Met$-enrichment: there is a forgetful functor $\Mod(\cat{L}) \to \Met$ which carries a model $M$ to $M(x)$.

The following result is \cite[Theorem 4.3]{power1999} specialized to our case. Let $\Law$ denote the category of $\Met$-Lawvere theories and $\ncat{CMnd}$ the subcategory of countable-arity monads and monad transformations.

\begin{proposition}
	Given a $\Met$ Lawvere theory $\cat{L}$, the category of its models is monadic over $\Met$.
	The monad arising from this adjunction is always of countable arity (it preserves countably filtered colimits).
	This construction induces a functor $\Law \to \Mnd(\CMet)$, which is an equivalence of categories. Its inverse carries a monad $T$ from $\Mnd(\CMet)$ to 
	the dual of its Kleisli category restricted to countable metric spaces.
	
\end{proposition}


\section{Metric equational theories}\label{sec:MET}

In this section, we revisit the theory of quantitative equational logic with the intention of making use of metric arities and produce explicit $\Met$-enriched Lawvere theories that we call \emph{metric equational theories (MET)}. To this end, we need to extend the syntax proposed in \cite{raduQAR2016} to properly encode operators with metric arities. We will have two kinds of judgement. To the structural judgements of type 
$\Gamma \vdash s =_\epsilon t$ from QET we will add \textit{formational judgments} of type $\Gamma \vdash t \ok$, asserting that in the given context, $t$ is a well-formed term -- i.e., it is provably well-defined. The \textit{formational atoms} of type $t \ok$ (for terms $t$) together with the \textit{quantitative equalities} of type $s=_\e t$ are the building blocks for our judgements.
We combine these types of judgment because they guarantee that the assertions we use involve well-formed terms, so we prove $t \ok$ before applying inference rules involving $t$. But to know which terms are well-formed, we reason about their distances.
A judgment that a term is well-defined ultimately boils down to a judgment about distances between the subterms -- as such, the theory could be rewritten without the formational judgments. However, they are a useful bookkeeping device and simplify the presentation.

From the start, we're faced with a difficulty not seen in the classical case. Since our operations have arities in metric spaces, their domain of definition depends on the distance between the arguments. This means that the set of terms under consideration is not the entire set of trees of operators as in the ordinary case, but only a subset. But, crucially, the set of well-formed terms depends not only on the context, but on the equational axioms.

While it would be possible to give mutually inductive definitions of ``theory'', ``sequent'' and ``term'', we find it simpler to consider the whole set of ``preterms'' and define the well-formed subset.

\begin{definition}
	A \emph{metric signature} $\Omega$ is a set of operation symbols, each equipped with a \emph{metric arity}, which is a countable cardinal 
	equipped with a metric. We write $f: N \in \Omega$ if $f$ has the arity $N = (N,d)$. 
	For a set $X$ of variables, $\tilde{\Omega}(X)$ is the set of \emph{preterms}, given inductively by:
	\begin{enumerate}
		\item For each $x\in X$, there is a preterm $x \in \tilde{\Omega}(X)$.
		\item If $t_i$ is a preterm for each $i \in N$, and $f: N \in \Omega$, there is a preterm $f(t_1, \dots )$.
	\end{enumerate}
\end{definition}
\begin{definition}
	A \emph{context $\Gamma$ over a set $X$ of variables} is a list of quantitative equalities $x =_\e y$, where $x,y\in X$ and $\e \in\Rp$.
	\\Given a metric space $M$ and an assignment $\alpha:X \to M$ of the variables, $\alpha$ satisfies the equation $x =_\e y$ iff $d(\alpha(x), \alpha(y)) \leq \e$.
	A variable assignment $\alpha: X \to M$ satisfies a context if it satisfies all the equations.
\end{definition}

\begin{definition}
	Fix a countably infinite set of variables $X$.
	A collection of judgments (formational and structural) is called a \emph{deducibility relation} if it is closed under the inference rules in Table \ref{tab:rules} stated for arbitrary $x,x_{i} \in X$, $f: (N,d) \in \Omega$, $~~t,t',t_{i},t_{i}',u_{i}\in \tilde{\Omega}(X)$, $\phi$ either a quantitative equality or a formational atom, and $\epsilon,\epsilon' \in \Rp$.
	\\Given a set $S$ of judgments, the least deducibility relation generated by them is denoted $\vdash_{S}$.
\end{definition}

\begin{table*}	
	\begin{align*}
		\begin{array}{c}
			\begin{aligned}
				\infer[\mathbf{Var}]{\Gamma \vdash x \ok}{}~~~~~
				&&&
				\infer[\mathbf{Assum}]{\Gamma \vdash x=_{\epsilon} y}{x =_{\epsilon} y \in \Gamma}~~~
				&&&
				\infer[\mathbf{Refl}]{\Gamma \vdash t =_{0} t}{\Gamma \vdash t \ok}~~~
				&&&
				\infer[\mathbf{Symm}]{\Gamma \vdash t =_{\epsilon} t'}{\Gamma \vdash t' =_{\epsilon} t}
			\end{aligned}
			\\\\
			\begin{aligned}
				\infer[\mathbf{Triang}]{\Gamma \vdash t =_{\epsilon+\epsilon'} t''}{\Gamma \vdash t=_{\epsilon} t', \Gamma \vdash t' =_{\epsilon'} t''}~~~
				&&&
				\infer[\mathbf{Max}]{\Gamma \vdash t =_{\epsilon} t'}{\epsilon' < \epsilon,\ \Gamma \vdash t =_{\epsilon'} t'}~~~
				&&&
				\infer[\mathbf{Cont}]{\Gamma \vdash t =_{\epsilon} t'}{\Gamma \vdash t =_{\epsilon'} t\ \forall \epsilon' > \epsilon}
			\end{aligned}
			\\\\
			\begin{aligned}
				\infer[\mathbf{Nexp}]{\Gamma \vdash f((s_i)) =_\epsilon f((t_i))}{\Gamma \vdash s_i =_{d(i,j)} s_j,\ \Gamma \vdash t_i =_{d(i,j)} t_j,\ \Gamma \vdash s_i =_\epsilon t_i \forall i,j}~~~~
				&&&
				\infer[\mathbf{Subst}]{\Gamma \vdash s[u_i/x_i] =_\epsilon t[u_i/x_i]}{\Gamma \vdash u_i =_{\delta_{ij}} u_j, \{x_{i} =_{\delta_{ij}} x_{j}\} \vdash s =_{\epsilon} t}
			\end{aligned}
			\\\\
			\begin{aligned}
				\infer[\mathbf{App}]{\Gamma \vdash f((t_{i})) \ok}{\forall i: \Gamma \vdash t_i \ok,\ \forall i,j, d(i,j)< \infty : \Gamma \vdash t_i =_{d(i,j)} t_j}
			\end{aligned}
		\end{array}
	\end{align*}	
	\caption{Deduction system}
	\label{tab:rules}
\end{table*}

Since the arities of the operations in a theory determine their domain of definition, we don't need formational judgements as \emph{axioms}.
However, when axiomatizing a theory, we need to ensure that for each axiomatic equation, the terms are well-formed.
But which terms are well-formed in a given context depends on the theory, since the theory provides bounds on distance that may imply well-formedness.
So, we have a kind of cyclic dependency, which makes things more difficult than they would be without metric arities.
However, in practice, the theory is not difficult to work with. First, since well-formedness of a term depends only on the distances between its subterms, there is no actual problem of cyclical dependency ---we can proceed from subterms to larger terms, forming judgments about their distances and well-definedness inductively.
And second, when working with a set of equational axioms which define the theory, it suffices to prove that those equations are well-formed ---then any sequents provable from them will also be well-formed.

\begin{definition}
	A collection of judgments $S$ is called \emph{well-formed} if, whenever $\Gamma \vdash t =_\epsilon s \in S$, we also have $\Gamma \vdash t \ok,~~ \Gamma \vdash s \ok \in S$.
	A well-formed collection of judgments $\eT$ is called a \emph{metric equational theory} (here just a \emph{theory}) if $\et{T} =\ \vdash_{\et{T}^=}$, where $\et{T}^= \subseteq \eT$ is the subset of equational judgments.
\end{definition}

\begin{proposition}\label{prop1}
	Let $S$ be a collection of equational judgments. Suppose that for each $\Gamma \vdash t =_\epsilon s \in S$, we have $\Gamma \vdash_S t \ok,~~ \Gamma \vdash_S s \ok$. Then $\et{T} =\ \vdash_S$ is a theory.
\end{proposition}

	Let $\Gamma$ be a context and let $\hat{X}_\Gamma$ be the set $X$ equipped with the pseudometric\footnote{This minimum is attained because of $\mathbf{Cont}$, and this is a pseudometric because of $\mathbf{Symm}, \mathbf{Triang}$ and $\mathbf{Refl}$.}
	$$d(x,x') = \min \{\epsilon \mid \Gamma \vdash x =_\epsilon x'\}.$$
	$X_\Gamma$ is the metric reflection of $\hat{X}_\Gamma$.
	Note that a nonexpansive map $\alpha:X_\Gamma \to M$ is precisely a variable assignment that satisfies $\Gamma$.

\begin{definition}
	Let $\Omega$ be a signature. A \emph{model} $M$ of $\Omega$ is a metric space $M$ such that for each $f: A \in \Omega$, there is a map $M[f]: M^A \to M$, where $M^A$ is the metric space of nonexpansive maps from $A$ to $M$ equipped with the supremum metric. 
	\\Given models $M, N$ of $\Omega$, a homomorphism is a nonexpansive map $\phi: M \to N$ so that $$\phi(M[f](x_1, \dots) = N[f](\phi(x_1), \phi(x_2), \dots).$$ 
\end{definition}

	Note that if $(x_i) \in M^A$, then $(\phi(x_i)) \in N^A$ since $\phi$ is nonexpansive. The category of models and homomorphisms is denoted $\Mod(\Omega)$.
	
	We will now define recursively what it means for $M$ to satisfy a judgment, and the interpretation $M[t]: M^{X_\Gamma} \to M$ of every term as a function, whenever $M$ satisfies $\Gamma \vdash t \ok$.
	
	\begin{enumerate}
		\item Every model satisfies every judgment $\Gamma \vdash x \ok$ for $x$ a variable. $M[x](\alpha)$ is simply $\alpha(x)$.
		\item $M$ satisfies $\Gamma \vdash f(t_1, \dots) \ok$ if it satisfies each $\Gamma \vdash t_i \ok$ and $d(M[t_i],M[t_j]) \leq d_f(i,j)$ for all $i,j$ (using the supremum metric on the function space). In this case $M[f(t_1, \dots)](\alpha) = M[f](M[t_1](\alpha), \dots)$.
		\item $M$ satisfies $\Gamma \vdash t =_\epsilon s$ if it satisfies both $\Gamma \vdash t \ok$ and $\Gamma \vdash s \ok$, and if $d(M[t], M[s]) \leq \epsilon$ in the function space.
	\end{enumerate}
	We write $\Gamma \vDash_M \phi$ if $M$ satisfies the sequent $\Gamma \vdash \phi$.
	
	\begin{definition}
	A model of a theory $\et{T}$ is a model of the signature of $\et{T}$ which satisfies every sequent in $\et{T}$.
	The category of models of $\et{T}$, is the full subcategory $\Mod(\et{T}) \subseteq \Mod(\Omega)$ spanned by the models of $\et{T}$.
	\end{definition}

The following proposition is readily verified:
\begin{proposition}[Soundness]
	Let $M$ be any model of $\Omega$. Then the relation $\vDash_M$ is well-founded and closed under the inference rules in Table \ref{tab:rules}. In particular, to prove that $M$ is a model of a theory $\et{T}$ generated by some axioms $S$, it suffices to prove that $M$ satisfies all the axioms.
\end{proposition}

	A theory is always defined over some signature $\Omega$. We will often just leave the signature implicit when speaking of a theory.
	When $\eT$ is a theory over $\Omega$ and $f$ is a symbol in $\Omega$, we will abuse notation by writing $f: N \in \eT$, speak of ``an operation in $\eT$'', and so on.


\section{Free models and monadicity}\label{sec:monad}
We now turn to the comparison of METs and Lawvere theories over $\Met$.
First, we will prove that for each MET $\eT$, the forgetful functor admits a left adjoint
(taking each metric space to a free model on it), and that this adjunction is monadic.
After proving that the monads are of countable arity, it follows by the general theory of enriched Lawvere theories (\cite{power1999}) that these monads (and hence the categories of models) come from enriched Lawvere theories.

First, we'll prove that all theories have initial models.
Then, given a theory $\eT$ and a metric space $A$, we construct a new theory
whose models are models of $\eT$ equipped with a map from $A$ (objects of the comma category $\Mod(\eT)_{A/}$).
The initial models of these theories are precisely the free models of $\eT$, and the fact that they all exist
proves the existence of a left adjoint.

\begin{proposition}\label{prop2}
	Let $\eT = (\Omega, \eT)$ be a metric equational theory.
	Consider the following metric space:
	\begin{itemize}
		\item its elements are the closed terms $t \in \tilde{\Omega}(\emptyset)$ such that $\vdash_\eT t \ok$, quotiented by the equivalence relation defined by $t \sim t' \Leftrightarrow ~~\vdash_\eT t =_0 t'$.
		\item for $t,t'\in\tilde\Omega(\emptyset)$,  $d([t],[t']) = \min \{\epsilon \mid \vdash_\eT t =_\epsilon t'\}$
	\end{itemize}
	Given an operation $f: A \in \Omega$ and a collection of elements $[t_i]$ satisfying $d([t_i],[t_j]) \leq d_A(i,j)$ for all $i,j \in A$, $[f((t_i)_{i \in A}]$ is another well-defined element, and this gives a model of $\eT$.
	This model is the initial model of $\eT$.
\end{proposition}

\begin{proof}
	We denote this metric space by $\Free^\eT(\emptyset)$ (clearly the initial model is the free model on the empty space -- we construct an entire functor $\Free^\eT$ later).
	
	Let $t, t', s$ be elements and suppose $\vdash t =_0 t'$, and $\vdash t =_\epsilon s$.
	Then using $\mathbf{Triang}$ and $\mathbf{Symm}$, also $\vdash t' =_\epsilon s$, and vice versa.
	Hence the distance is well-defined on equivalence classes. $\mathbf{Refl}, \mathbf{Symm}, \mathbf{Triang}$ straightforwardly imply that it's an (extended) metric.
	(Note that the minimum defining $d$ is always attained, by $\mathbf{Cont}$).
	
	$\mathbf{App}$ implies that applying functions to a collection of well-formed closed terms with suitable bounds on their distance results in another well-formed closed term.
	$\mathbf{NExp}$ implies that this is well-defined (if we replace each input term with an equivalent one, the resulting terms are equivalent) and nonexpansive, so these operations $\Free^\eT(\emptyset)$ of $\Omega$.
	
	Given some equation $\Gamma \vdash s =_\epsilon t$ in $\eT$ using the variables $\{x_i\}$, and elements $u_i \in \Free^\eT(\emptyset)$ so that $d([u_i],[u_j]) \leq \epsilon$ whenever $x_i =_\epsilon x_j \in \Gamma$ (in other words, a variable assignment satisfying $\Gamma$), by $\mathbf{Subst}$ this variable assignment also satisfies $s =_\epsilon t$. Hence, this is a model of $\eT$.
	
	Now suppose $M$ is another model of $\eT$.
	Note that if $t$ is a term without variables, $M[t]: M^{X_\Gamma} \to M$ is constant (by induction on $t$) - let's abuse notation by writing $M[t] \in M$ for the constant value of this map.
	(If $M$ is empty, of course, this definition won't make sense. But in that case $\eT$ must have no constant symbols, and so there can't be any terms without variables).
	Define a map $\phi: \Free^\eT(\emptyset) \to M$ by sending each class $[t]$ into $M[t]$.
	Since $M$ satisfies $\eT$, if $[t] = [t']$, then $d_M(M[t],M[t']) = 0$, so they are equal. Hence, this is well-defined.
	Analogously, if $d(t,t') = \epsilon$, then $d_M(M[t],M[t']) \leq \epsilon$, so this is a nonexpansive map.
	It's clearly a homomorphism.
	On the other hand, clearly the value of any homomorphism on closed terms is determined ---it must go to the interpretation of that term. So this is unique.
\end{proof}

\begin{definition}
	Let $A$ be a metric space.
	Then $\Omega^A$ is the signature with one nullary operation $[a]$ for each point in $A$,
	and $\eT(A)$ is the theory over this signature generated by the sequents $\vdash [a] =_{d(a,a')} [a']$ for each $a,a' \in A$.
\end{definition}

Recall that, given a functor $F: \cat{C} \to \cat{D}$ and an object $A \in \cat{D}$, the comma category $\cat{C}_{A/}$ has objects
pairs $(X \in \cat{C}, f: A \to FX)$, and morphisms from $(X,f) \to (Y, g)$ given by $\phi: X \to Y \in \cat{C}$ so that $gF(\phi) = f$ (i.e so that the obvious triangle in $\cat{D}$ commutes).

\begin{proposition}\label{prop3}
	As categories over $\EMet$, $\Mod(\eT(A)) = \Met_{/A}$.
\end{proposition}

	Let $\eT, \eT'$ be two theories. Then, we denote by $\eT \sqcup \eT'$ the disjoint union of the two theories, which is the theory with signature $\Omega^\eT \sqcup \Omega^{\eT'}$,
	and generated by the union of the sequents in $\eT$ and $\eT'$.

\begin{proposition}\label{prop4}
	\[\Mod(\eT \sqcup \eT') \cong \Mod(\eT) \times_\Met \Mod(\eT').\]
	In other words, to give a space the structure of a model of $\eT \sqcup \eT'$ is simply to give it independently the structure of a model of each theory; and a homomorphism is just a function which is a homomorphism for each theory separately.
\end{proposition}

Applying the preceding proposition to the characterization of $\Mod(\eT(A))$, we obtain the following:

\begin{corollary}
	$\Mod(\eT \sqcup \eT(A)) \cong \Mod(\eT)_{/A}$.  
\end{corollary}
Since a left adjoint exists if and only if each comma category has an initial object (and is then given by these initial objects, see eg \cite[Theorem IV.1.1]{MacLane}), we obtain:
\begin{corollary}
	$\Mod(\eT) \to \Met$ admits a left adjoint, which we denote $\Free^\eT(A)$.
	Concretely, $\Free^\eT(A)$ is given by terms $t$ using the operations of $\eT$ and a further constant symbol for each element of $A$, quotiented by provable equality, with $d(t,t')$ being the smallest $\epsilon$ so that $t=_\epsilon t'$ is provable using $\eT$ and the further axioms $\vdash [a] =_{d(a,a')} [a']$ for every pair of elements in $A$.
\end{corollary}

We now turn to the proof of monadicity ---this is really the key ingredient in the comparison with Lawvere theories.
Of course, the existence of the left adjoint $\Free^\eT$ (along with the proof that the associated monad has countable arity, see below) already gives us a Lawvere theory ---the question is whether its models are the same as the models of $\eT$.

\begin{proposition}
	The adjunction $\Mod(\eT) \to \Met$ is monadic.
\end{proposition}
\begin{proof}
	We will apply Beck's monadicity theorem (see eg \cite[Theorem 7.1]{MacLane}.
	We must prove that the forgetful functor has a left adjoint, that both categories are finitely complete, and that
	the forgetful functor creates coequalizers for those pairs which have split coequalizers in $\Met$.
	
	The first condition we already proved, it's easy to see that the forgetful functor creates limits, taking care of the completeness.
	So let's look at the last condition.
	Let $(l,r): R \to M$ be a pair of homomorphisms in $\Mod(\eT)$.
	A split coequalizer of the underlying metric spaces of this diagram is a diagram:
	
	\begin{center}
		\begin{tikzcd}[column sep=5em]
			R \ar[r, shift left=1, "l"] \ar[r, shift right=1, swap, "r"]& M \ar[l, bend right=30, swap, "t"] \ar[r, swap, "e"] & Q \ar[l, bend right=30, swap, "s"]
		\end{tikzcd}
	\end{center}
	so that $es= 1_Q, se = rt, lt = 1_M$.
	Note that in this case,
	\[d_Q(q,q') = d_Q(es(q),es(q')) \leq d_M(s(q), s(q')) \leq d_Q(q,q'),\]
	so $d_Q(q,q') = d_M(s(q),s(q'))$, and $e$ is surjective with $e(m) = e(m')$ if and only if
	$rt(m) = rt(m')$.
	This implies that $Q$ as a set is the coequalizer of $l,r: R \to M$,
	and for each point $m \in M$, we can pick out the pair $t(m) \in R$ which identifies $m$ and $s(e(m))$.
	
	We must show that we can equip $Q$ with the structure of an $\eT$-model, so that $e$ becomes a homomorphism and so that $Q$ acquires the universal property of coequalizing $(l,r)$.
	
	First, let's define $Q[f]$ for each operation symbol $f \in \eT$.
	Since $s$ is distance nonexpanding, we can certainly define $Q[f](q_1, \dots)$ to be $e(M[f](s(q_1), \dots))$.
	
	Now, suppose $\Gamma \vdash t =_\epsilon s$ is an equation in $\eT$.
	Fix a variable assignment $\alpha: X \to Q$ satisfying $\Gamma$.
	Then postcomposing with $s$ gives an assignment in $M$ which also satisfies $\Gamma$ (since this is nonexpansive).
	It's apparent by structural induction on $t$ that $Q[t](\alpha)$, the value of $t$ under $\alpha$, equals $e(M[t](s\alpha)$. Since $M$ is a model and $e$ is nonexpansive, the equation also holds for $Q$. Hence, $Q$ is a model.
	
	Now, if $(m_1, \dots)$ is an element of $M^N$ (where $N$ is the arity of $f$),
	$(t(m_1), t(m_2), \dots)$ gives an element of $R^N$, so that $R[f](t(m_1), \dots)$ must identify $M[f](m_1, \dots)$ and $M[f](se(m_1), \dots)$.
	This means $Q[f](e(m_1), \dots) = e(M[f](se(m_1), \dots)) = e(M[f](m_1, \dots))$ (since $e$ equalizes $l$ and $r$), so that $e$ is a homomorphism for the $\eT$-structures on $M$ and $Q$.

	Now suppose given a homomorphism $\phi: M \to A$ which equalizes $l$ and $r$, for some other model $A$.
	We may attempt to define $Q \to A$ as the composite $Q \overset{s}{\to} M \to A$.
	This will be nonexpansive, and it's straightforward to see that it'll be a homomorphism.
	On the other hand, if $\hat{\phi}: Q \to A$ is any homomorphism so that $\hat{\phi} e = \phi$, we have
	$\hat{\phi} = \hat{\phi}  e r = \phi r$, so this is in fact the unique such morphism.
	This concludes the proof of monadicity.
\end{proof}
\begin{proposition}
	The functor $\Free^\eT: \Met \to \Met$ is countable-arity, i.e it commutes with countably filtered colimits.
\end{proposition}
\begin{proof}
	Let $(X_i)_{i \in I}$ be a diagram in $\Met$ over a countably filtered index category $I$.
	We must prove
	\[\colim_i \Free^\eT(X_i) \to \Free^\eT(\colim_i X_i)\]
	is an isomorphism.
	Let $f: A = (A,d)$ be an operation symbol in $\eT$.
	Note that countable powers commute with countably filtered colimits in $\Met$ (this is essentially what it means that the countable metric spaces are the countably small objects in $\Met$),
	so
	\[(\colim_i \Free^\eT(X_i))^A \cong \colim_i \Free^\eT(X_i)^A\]
	Hence, we can equip the colimit with the structure of a model of the signature of $\eT$,
	by defining the interpretation of $f$ to be the map
	\[(\colim_i \Free^\eT(X_i))^A \cong \colim_i \Free^\eT(X_i)^A \to \colim_i \Free^\eT(X_i),\]
	where the last map just applies $f$ in each part of the colimit.
	This amounts to, given a sequence $t_a$, finding $X_i$ so that they're all present in $\Free^\eT(X_i)$ (and at the right distance) and just taking $f((t_a))$.
	
	Given an equation $\Gamma \vdash t =_\epsilon s$ in $\eT$, let $\alpha: X \to \colim_i \Free^\eT(X_i)$ be some assignment validating $\Gamma$.
	Then there is some $i$ so that this factors as a map $X \to \Free^\eT(X_i)$ which also validates $\Gamma$.
	Up to equivalence, the interpretation of $t$ and $s$ under this assignment in the colimit is just their interpretation under this assignment in $\Free^\eT(X_i)$.
	Since this free model satisfies the equation, and the inclusion of this in the colimit is nonexpansive, the colimit also satisfies the equation. Hence, it's a model of $\eT$.
	
	The maps $X_i \to \Free^\eT(X_i)$ induce a map \[\colim_i X_i \to \colim_i \Free^\eT(X_i),\] which, since the latter is a model, induces a map
	\[\phi: \Free^\eT(\colim_i X_i) \to \colim_i \Free^\eT(X_i)\]
	By initiality, the composite 
	\[\Free^\eT(\colim_i X_i) \to \colim_i \Free^\eT(X_i) \to \Free^\eT(\colim_i X_i)\]
	is the identity, proving $\phi$ is an isometry.
	To see it's surjective, consider some element in $\colim_i \Free^\eT(X_i)$.
	It's represented by some $t \in \Free^\eT(X_i)$ for some $i$,
	i.e some term using the operations of $\eT$ and constant symbols from $X_i$.
	Let $t'$ be a term obtained from $t$ by replacing each $[x]$ with $[\bar{x}]$, where $\bar{x} \in \colim_i X_i$ is the equivalence class of $x$.
	Then up to the equivalence relation in the colimit, $\phi(t') = t$.
\end{proof}
\begin{corollary}\label{cor:METisLaw}
	For any equational theory $\eT$, there is a $\Met$-Lawvere theory with an equivalent category of models.
\end{corollary}


\section{Equational theories from Lawvere theories}\label{sec:LawisMET}
	Let $\lT$ be a $\Met$ Lawvere theory.
	Fix a choice of distinct variables $x_1, x_2, \dots$.
	Given a metric cardinal\footnote{A metric cardinal is a cardinal endowed with a metric; hence we might have different metric cardinals supported by the same set.} $C$, we let $\Gamma(C)$ be the context $\{x_i =_{d_C(i,j)} x_j \mid i,j \in C\}$.
	Then we can define a theory $\eT(\lT)$, as follows:
	\begin{enumerate}
		\item The signature $\Omega^\lT$ has a symbol $[f]: C$ for every countable metric cardinal $C$ and $f \in \lT(x^C,x)$.
		\item Given $f, f' \in \lT(x^C,x)$, there is an axiom $\Gamma(C) \vdash [f](x_1, \dots) =_{d_{\lT(C)}(f,f')} [f'](x_1, \dots).$ 
		\item For $i \in C$, let $\pi_i: x^C \to x$ be the projection to the $i$th component. Then we have an axiom
		$$\Gamma(C) \vdash [\pi_i](x_1, \dots) =_0 x_i$$
		\item Given $f: x^C \to x$ and a tuple $(g_1, \dots) \in \lT(x^D,x)^C$, we can compose these using the isomorphism $\lT(x^D,x)^C = \lT(x^D,x^C)$. Denote this composition as $f \circ (g_1, \dots)$.
		Then we have an axiom $$\Gamma(D) \vdash [f]([g_1](x_1, \dots),[g_2](x_1,\dots), \dots) = [f \circ (g_1, \dots)](x_1, \dots)$$
	\end{enumerate}

\begin{proposition}\label{prop:LawisMET}
	$\Mod(\eT(\lT)) \cong \Mod(\lT)$ as categories over $\Met$.
	In particular, a strong monad on $\Met$ comes from a metric equational theory if and only if it is of countable arity.
\end{proposition}
\begin{proof}
	Let $M$ be a model of $\lT$. Then we can equip $M(x)$ with the structure of a model of $\eT(\lT)$ as follows:
	For each operation $f \in \lT(x^C,x)$, we define $M(x)[f] : M(x)^C \to M(x)$ simply as $M(f)$, using the canonical isomorphism $M(x)^C \cong M(x^C)$.
	Since $M$ is assumed to be an enriched functor, the map $\lT(x^C,x) \to \Met(M(x)^C,M(x))$ is nonexpansive, and hence this structure satisfies the axioms from part 2 of the definition.
	Since $M$ preserves powers, it carries the projections to the projections, so it satisfies part 3.
	And since $M$ is functorial, it preserves the composition, which means it satisfies part 4.
	
	Hence, this is a model of $\eT(\lT)$.
	If $\phi: M \to N$ is a natural transformation, the induced map $\phi_{x^C}: M(x)^C \to N(x)^C$ is given by the $C$-fold power of $\phi_x$.
	Since $\phi$ is natural, this implies that $\phi_x$ is a $\eT(\lT)$-homomorphism, so this construction defines a functor, which we want to show is an isomorphism.
	Since $\Mod(\lT)$ is equivalently the Eilenberg-Moore category of the associated monad, the forgetful functor $\Mod(\lT) \to \Met$ is an isometry on morphisms. This implies that the functor $\Mod(\lT) \to \Mod(\eT)$ we've just constructed is an isometry on morphisms as well, so we have to show that it's full and essentially surjective.
	For fullness, it's clear that being a homomorphism between $M(x)$ and $N(x)$ requires commuting with all the operations in $\lT$, which is just what it means to be a natural transformation.
	For essential surjectivity, let $M$ be some model of $\eT(\lT)$.
	Define $M(x^C) = M^C$, and given a family $(f_a \in \lT(x^C, x))_{a \in A}$ representing a morphism $f: x^C \to x^A$, let $M(f)(m_1, \dots) = (M[f_a](m_1, \dots))_{a \in A}$.
	By axiom 2., if these operations satisfy the distance bounds $d(f_a, f_{a'}) \leq d(a,a')$, then so will the resulting map $M(x)^C \to M(x)^A$,
	and by axioms 3. and 4., this is functorial and preserves powers, so it defines a model that goes to $M$.
\end{proof}

\section{Completeness}\label{sec:Completeness}
In this section, we will prove the following completeness theorem for our theory.

\begin{theorem}\label{thm1}
	Let $\eT$ be a theory, and let $\Gamma \vdash \phi$ be a sequent.
	Suppose every model of $\eT$ satisfies this sequent. Then $\Gamma \vdash_\eT \phi$.
\end{theorem}
We will need the following characterization of $\eT \sqcup \eT(A)$ for $A$ a countable metric space:

\begin{lemma}\label{lemma1}
	Let $A$ be any metric space.
	Consider a sequent $\Gamma \vdash \phi$ in the signature of $\eT \sqcup \eT(A)$.
	We form the sequent $\Gamma^A \vdash \phi^A$ in the signature of $\eT$ by the following procedure:
	\begin{enumerate}
		\item If necessary, relabel some of the variables so that there is an infinite set of unused variables.
		\item Since only countably many of the constant symbols $[a],~ a \in A$ can occur in $\phi$, choose a distinct fresh variable $x_a$ for each of these. Let $\phi^A$ be $\phi$ with each occurrence of $[a]$ replaced by $x_a$.
		\item Let $\Gamma^A$ be $\Gamma \cup \{x_a =_{d_A(a,a')} x_{a'} \mid a,a' \in A, d(a,a') < \infty\}$. 
	\end{enumerate}
	Note that some arbitrary choices are involved in defining $\Gamma^A$ and $\phi^A$.
	Regardless of the choices, $\Gamma^A \vdash_\eT \phi^A$ if and only if $\Gamma \vdash_{\eT \sqcup \eT(A)} \phi$.
\end{lemma}

With this lemma, we can prove the theorem.

\begin{proof}[Proof of Theorem \ref{thm1}]
	Consider the tautological variable assignment, given by the identity function $X \to X_\Gamma$.
	By composing with the unit, we get a variable assignment $X \to X_\Gamma \to \Free^\eT(X_\Gamma)$.
	Note that by construction, this satisfies the hypotheses of $\Gamma$.
	Hence, by assumption, it must satisfy $\phi$.
	But this means that $\vdash_{\eT \sqcup \eT(X_\Gamma)} \phi[[x]/x]$, which by Lemma \ref{lemma1}, means that
	$\Gamma(X_\Gamma) \vdash_\eT \phi$, implying $\Gamma \vdash_\eT \phi$.
\end{proof}


\section{Cauchy completion}\label{sec:Cauchy}
\label{sec:CC}
Many important metric spaces are given as completions of more simply defined subspaces.
For example, the reals are the completion of the rationals, the $L^p(X,\mu)$ spaces are the completion of the continuous functions on $X$ (in the $L^p$-metric), and so on.
Similarly, in \cite{raduQAR2016}, it is proven that the space of probability measures on $X$ is constructed as the completion of $\Free^\eT(X)$ for a certain theory $\eT$, and the space of closed subsets is the completion of the free semilattice on $X$ (assuming $X$ is compact).

Completion already interacts well with QETs, because the operation of completion preserves (finite) products.
This means that if $\eT$ is a theory (with finitary operations) and $M$ is a model, the completion $\overline{M}$ has a canonical model structure making $M \to \overline{M}$ a homomorphism, and this is universal among homomorphisms from $M$ to complete models.
So $\overline{\Free^\eT(X)}$ has a good universal property.

However, with our expanded theory, we can do even better.
The monad $\overline{(-)}: \Met \to \Met$, which carries a metric space to its completion, is of countable arity, and so it is represented by a metric equational theory.
There are obviously many distinct axiomatizations of this monad. Here is an example:

	Let $N$ be the natural numbers equipped with the metric $d(n,m) = \frac{1}{2^{\min(n,m)}}$ for $n \neq m$, and $0$ if they're equal.
	Let $\eT^\mathrm{comp}$ be a theory with one operation $\lim:N$,
	and the axioms
	$$\{x_n =_{d(n,m)} x_m \mid n,m \in N\} \vdash \lim(x_1, \dots) =_{1/2^n} x_n.$$
	
\begin{proposition}\label{prop5}
	Models of $\eT^\mathrm{comp}$ are precisely complete metric spaces, $M[\lim]: M^N \to M$ always carries a sequence to a limit of that sequence, and every nonexpansive map between models is a homomorphism.
	In particular, $\Mod(\eT^\mathrm{comp})$ is equivalent to the full subcategory of complete metric spaces, and the free model on $X$ is the completion.
\end{proposition}

\begin{corollary}
	Let $\eT$ be any theory.
	Then $\Mod(\eT \sqcup \eT^\mathrm{comp})$ is equivalent to the full subcategory of $\Mod(\eT)$ spanned by the complete models.
\end{corollary}
Thus, for example, by taking the disjoint union of $\eT^\mathrm{comp}$ with the theory of $p$-interpolative barycentric algebras from \cite[section 10]{raduQAR2016}, we get a theory whose free model on a separable metric space $X$ is the space of Borel probability measures on $X$ in the $p$-Wasserstein metric.


\section{Quantitative equational theories as metric equational theories}\label{sec:Specials}

Given a signature where all the arities are discrete, every sequent
$\vdash t \ok$ is trivially provable by repeated application of $\mathbf{App}$ and $\mathbf{Var}$.
In this case, a metric equational theory over this signature is simply a quantitative equational theory,
the notion of model is the same, etc.

It is interesting to ask which monads $T: \Met \to \Met$ are axiomatizable by quantitative equational theories.
Using our equivalence between metric equational theories and Lawvere theories over $\Met$, we can answer this question.

\begin{theorem}\label{surj-disc}
	A countable-arity monad $T$ is axiomatizable by a quantitative equational theory $\eT$, if and only if $T$ preserves surjections, i.e. $T(f): TX \to TY$ is a surjection whenever $f: X \to Y$ is a surjection.    
\end{theorem}

\begin{proof}
	Let $\eT$ be any quantitative equational theory, viewed as a metric equational theory.
	That is, it is an MET whose operations all have discrete arities.
	Clearly, for any metric space $A$, also $\eT(A)$ and hence $\eT \sqcup \eT(A)$ have this property.
	Then $\vdash_{\eT \sqcup \eT(A)} t \ok$ for \emph{any} term $t$.
	Let $f: A \to B$ be a surjection of metric spaces.
	The induced map $\Free^\eT(A) \to \Free^\eT(B)$ is given by replacing each constant symbol $[a]$
	in a term $t$ in the theory $\eT \sqcup \eT(A)$ with $[f(a)]$.
	Since $f$ is surjective, given a term $t$ in $\eT \sqcup \eT(B)$, we can always find $t'$ which would be mapped to it by this procedure. And since all terms are well-formed in $\eT \sqcup \eT(A)$, this $t'$
	represents an element of $\Free^\eT(A)$ which is therefore in the preimage of $t$.
	So the map $\Free^\eT(f)$ is surjective.
	
	Conversely, suppose $T$ is of countable arity and preserves surjections.
	Given a metric space $A$, let $A^d$ be the underlying set of $A$ equipped with the discrete metric $d(a,b) = \infty$. Note that there is a surjection (in fact a bijection) $A^d \to A$. 
	
	$T$ can be axiomatized by $\eT(\lT(T))$, but this contains a number of operation symbols of non-discrete arity.
	For each operation $f: x^A \to x$ where $A$ is not discrete, we can choose a factorization over $x^A \to x^{A^d}$,
	because the precomposition map $\lT(T)(x^{A^d}, x) \to \lT(x^A,x)$ is isomorphic to $T(A^d) \to T(A)$ and hence surjective. Choose such a factorization $\bar{f}: x^{A^d} \to x$ for each $f$.
	Note that $$\Gamma(A) \vdash_{\eT(\lT(T))} [f](x_1, \dots) = [\bar{f}](x_1, \dots).$$
	Let us abbreviate $\eT = \eT(\lT(T))$.
	Now consider a theory $\eT'$ defined as follows:
	\begin{enumerate}
		\item The operation symbols are the \emph{discrete} operations of $\eT(\lT(T))$.
		\item Whenever $\Gamma \vdash s =_\epsilon t \in \eT$, we let $\Gamma \vdash \bar{s} =_\epsilon \bar{t} \in \eT'$, where $\bar{s}, \bar{t}$ denote the result of replacing each occurrence of an operation symbol $[f]$ in $s$ and $t$ with $[\bar{f}]$.
		\item For any term, $\Gamma \vdash t \ok$ is in $\eT'$.
	\end{enumerate}
	This set of sequents is clearly stable under the inference rules (since $\eT$ is).
	The signature of $\eT'$ is a subset of the signature of $\eT$, so any term in the former is also a term in the latter.
	Note that if $\Gamma \vdash_{\eT'} s =_\epsilon t$, then there are some terms $s', t'$ in $\eT$ so that
	$\Gamma \vdash_{\eT} s' =_\epsilon t'$ and $\bar{s'} = s, \bar{t'}=t$.
	But by repeatedly using the equation $\Gamma(A) \vdash [f](x_1, \dots) = [\bar{f}](x_1, \dots)$, (where $A$ is the arity of $f$), we find that we must have $\Gamma \vdash s' = s$ and similarly for $t' = t$, and hence we must also have
	$\Gamma \vdash_\eT s =_\epsilon t$.
	
	Hence we get a natural forgetful functor $U: \Mod(\eT) \to \Mod(\eT')$ over $\Met$.
	It suffices to show this is an equivalence of categories.
	It is clearly faithful, because the composite $\Mod(\eT) \to \Mod(\eT') \to \Met$ is faithful.
	Given two models $M, N \in \Mod(\eT)$, and a homomorphism $\phi: UM \to UN$,
	note that
	\[\phi(M[f](x_1, \dots)) = \phi(M[\bar{f}](x_1, \dots)) = N[\bar{f}](\phi(x_1), \dots) = N[f](\phi(x_1), \dots),\]
	so $\phi$ is already a homomorphism $M \to N$. Hence, the forgetful functor is full.
	
	Finally, given a model $M'$ of $\eT'$, defining $M$ with the same underlying metric space and $M[f] = M'[\bar{f}]$ gives a model of $\eT$ with $UM = M'$.
	Hence, the functor is an equivalence of categories.   
\end{proof}

We can also ask which monads correspond to theories with only finitary operations.
Since $\Met$ is not locally finitely presentable, there is no clean correspondence between finitary monads and finitary Lawvere theories. And for similar reasons, many plausible characterizations of the monads presented by ``finitary METs'' fail. We will give a few counterexamples to demonstrate the problem.

\begin{example}
	Let $X_i$ be the set consisting of two points $a$ and $b$ at distance $1+ 1/i$,
	considered as a diagram indexed by the category $(\mathbb{N},\leq)$ (with identities as the structure morphisms).
	Then $\colim_i X_i$ consists of two points at distance $1$.
	
	\begin{enumerate}
		\item Consider the theory $\eT_1$ with one binary operation symbol $f: (\{x,y\}, d(x,y)=1)$. Then $\Free^{\eT_1}(X_i) \cong X_i$, but $\Free^{\eT_1}(\colim_i X_i)$ has three points, $a, b$ and $f(a,b)$.
		\item Consider the theory $\eT_2$ with no operations, and one equation $x=_1 y \vdash x =_0 y$.
		Then $\Free^{\eT_2}(X_i) \cong X_i$, but \[\Free^{\eT_2}(\colim_i X_i) = \{*\}.\]
	\end{enumerate}
	
	Thus, neither of these theories axiomatize monads which are finitary, in the sense that they commute with (finitely) filtered colimits.
	
	The problem in both cases is that quantitative equations ---whether as preconditions for the application of an operation, or preconditions for another equation--- are not ``finitary'', do not commute with finitely filtered colimits. An analogous problem prevents them from being ``strongly finitary'' in the sense studied in \cite{adamek2022quantitative}.
\end{example}

This problem seems to depend in an essential way on the \emph{discontinuity} of these conditions. The equation $x=_0y$ appears ``suddenly'' once $x=_1 y$.
Thus, for example, we can consider the ``theory of contractions'', having one unary symbol $s$ and the equations $x=_{2 \epsilon} y \vdash sx =_\epsilon sy$ for every $\epsilon$.
This is finitary in the sense that the associated monad (which is simply $X \mapsto X \times \mathbb{N}$ equipped with the metric $d((x,n),(y,n)) = d(x,y)2^{-n}, d((x,n),(y,m) = 0$ if $n \neq m$) commutes with filtered colimits.

In \cite{adamek2022quantitative}, recognizing essentially this problem, the notion of \emph{strongly finitary functor} $\Met \to \Met$ is studied. These are functors $F$ which equal the enriched left Kan extension of their restriction to the subcategory of finite and discrete metric spaces.
This intuitively corresponds to allowing only finite and discrete-arity operations, and allowing only axioms of the form $\vdash t =_\epsilon s$. When restricted to ultrametric spaces, this is indeed the case ---strongly finitary monads are exactly those presented by QETs of this form (\cite[Theorem 5.8]{adamek2022quantitative}).

However, this correspondence does \emph{not} hold, as the following example shows:

\begin{example}\label{counterexample-strongfinit}
	Let $\eT$ be the theory with one binary operation $f$, two unary operations $g,g'$, and the axiom $g =_1 g'$.
	Then $\Free^\eT: \Met \to \Met$ is not strongly finitary.
	
	To see this, recall that the left Kan extension under consideration is given by the coend formula
	\[\Lan_{i: \ncat{Fin} \to \Met}\Free^\eT(X) = \int^{F}X^F \otimes \Free^\eT(F)\]
	This means the question is whether the map
	\[\int^{F}X^F \otimes \Free^\eT(F) \to \Free^\eT(X),\]
	which carries a pair $(\alpha: F \to X, t \in \Free^\eT(F))$ to $\Free^\eT(\alpha)(t)$,
	is an isomorphism for all $X$.
	
	Let $X = \{x_1, x_2,x_3\}$ with $d(x_1, x_2) = 1, d(-,x_3) = \infty$.
	Consider the two terms $f([x_1],g([x_3]))$ and $f([x_2],g'([x_3]))$ in $\Free^\eT(X)$.
	Applying the axioms $\vdash [x_1] =_1 [x_2]$ and $\vdash g(y) =_1 g'(y)$, and nonexpansiveness, clearly these have distance at most $1$.
	
	Now consider the points $(\{a,b\}, \alpha, [f([a], g([b]))]), (\{a,b\}, \beta, [f([a],g'([b]))])$ in the coend,
	where $\alpha(a) = x_1, \alpha(b) = x_3, \alpha'(a) = x_2, \alpha'(a) = x_3$.
	Clearly $d_{X^{\{a,b\}}}(\alpha, \beta) = 1$, and also 
	\[d([f([a], g([b]))], [f([a],g'([b]))]) = 1.\]
	Hence, the distance of these points in the tensor product is $2$. (It is not too difficult to see that the distance in the coend, which is a quotient of the coproduct of all these tensor products, is not less than $2$).
\end{example}

\section{Conclusions and Related and Future Work}

Approximation is fundamental in mathematics and computer science and motivates the extension of universal algebra from the category $\Set$ to the category $\Met$. There are two current approaches - Quantitative Equational Theories and Enriched Lawvere Theories, but neither is a complete answer. QETs produce sound and complete systems for a notion of algebraic structure, but that notion is too weak to cover key examples such as Cauchy Completion. On the other hand, Enriched Lawvere Theories provide the right theoretical framework but are not accompanied by sound and complete proof systems needed to establish distances between terms in specific theories. This paper takes the best of each approach, producing what we call Metric Equational Theories. The fundamental innovation is the inclusion of metric arities for operators (motivated by the Enriched Lawvere Theory framework) within METs. 

There are a number of directions of future work, and we highlight some here. Firstly, Enriched Lawvere Theories don't give sound and complete proof systems for the equations of a theory. But in the case of metric spaces, we showed such systems exist. Can we find conditions under which such systems exist for a broad class of Enriched Lawvere Theories? Secondly, going beyond equational theories, we might ask what stronger systems look like. For example, equational theories correspond to finite product theories, but there is a natural notion of finite limit theory. The question is thus, how do we extend this paper to develop finite limit theories for metric spaces? Thirdly, and from a different perspective, algebraic theories underpin effectful programming languages. So how can we use the work contained in this paper to create effectful programming language constructs for approximate computation?

\textbf{Related work.}
In \cite{adamek2022quantitative} Adámek et.al study a subclass of QETs and prove a classification theorem quite analogous to ours in the context of \emph{ultrametric spaces}, showing that they present exactly the \emph{strongly finitary} monads on ultrametric spaces (\cite[Theorem 5.8]{adamek2022quantitative}). However, their classification does not extend to monads with operations of metric arity.
There is also recent work describing the category theory of $\Met$ from other points of view, e.g. \cite{AdamekHausdorff} describing \textit{Hausdorff polynomial functors} on $\Met$ and their associated monads, or \cite{Adamek2022} developing various species of \textit{approximate limits} in $\Met$.

\nocite{*}
\bibliographystyle{eptcs}
\bibliography{bibliography}

\begin{thebibliography}{10}
\providecommand{\bibitemdeclare}[2]{}
\providecommand{\surnamestart}{}
\providecommand{\surnameend}{}
\providecommand{\urlprefix}{Available at }
\providecommand{\url}[1]{\texttt{#1}}
\providecommand{\href}[2]{\texttt{#2}}
\providecommand{\urlalt}[2]{\href{#1}{#2}}
\providecommand{\doi}[1]{doi:\urlalt{https://doi.org/#1}{#1}}
\providecommand{\eprint}[1]{arXiv:\urlalt{https://arxiv.org/abs/#1}{#1}}
\providecommand{\bibinfo}[2]{#2}

\bibitemdeclare{book}{Adamek1994}
\bibitem{Adamek1994}
\bibinfo{author}{J.~\surnamestart Adamek\surnameend} \&
  \bibinfo{author}{J.~\surnamestart Rosicky\surnameend} (\bibinfo{year}{1994}):
  \emph{\bibinfo{title}{Locally Presentable and Accessible Categories}}.
\newblock \bibinfo{publisher}{Cambridge University Press},
  \doi{10.1017/cbo9780511600579}.

\bibitemdeclare{inproceedings}{AdamekHausdorff}
\bibitem{AdamekHausdorff}
\bibinfo{author}{Jir{\'{\i}} \surnamestart Ad{\'{a}}mek\surnameend},
  \bibinfo{author}{Stefan \surnamestart Milius\surnameend} \&
  \bibinfo{author}{Lawrence~S. \surnamestart Moss\surnameend}
  (\bibinfo{year}{2023}): \emph{\bibinfo{title}{On Kripke, Vietoris and
  Hausdorff Polynomial Functors ((Co)algebraic pearls)}}.
\newblock In \bibinfo{editor}{Paolo \surnamestart Baldan\surnameend} \&
  \bibinfo{editor}{Valeria \surnamestart de~Paiva\surnameend}, editors:
  {\slshape \bibinfo{booktitle}{10th Conference on Algebra and Coalgebra in
  Computer Science, {CALCO} 2023, June 19-21, 2023, Indiana University
  Bloomington, IN, {USA}}}, {\slshape \bibinfo{series}{LIPIcs}}
  \bibinfo{volume}{270}, \bibinfo{publisher}{Schloss Dagstuhl - Leibniz-Zentrum
  f{\"{u}}r Informatik}, pp. \bibinfo{pages}{21:1--21:20},
  \doi{10.4230/LIPICS.CALCO.2023.21}.

\bibitemdeclare{misc}{adamek2022quantitative}
\bibitem{adamek2022quantitative}
\bibinfo{author}{J.~\surnamestart Adámek\surnameend},
  \bibinfo{author}{M.~\surnamestart Dostál\surnameend} \&
  \bibinfo{author}{J.~\surnamestart Velebil\surnameend} (\bibinfo{year}{2022}):
  \emph{\bibinfo{title}{Quantitative Algebras and a Classification of Metric
  Monads}}.
\newblock \eprint{2210.01565}.

\bibitemdeclare{article}{Adamek2022}
\bibitem{Adamek2022}
\bibinfo{author}{J.~\surnamestart Adámek\surnameend} \&
  \bibinfo{author}{J.~\surnamestart Rosický\surnameend}
  (\bibinfo{year}{2022}): \emph{\bibinfo{title}{Approximate injectivity and
  smallness in metric-enriched categories}}.
\newblock {\slshape \bibinfo{journal}{Journal of Pure and Applied Algebra}}
  \bibinfo{volume}{226}(\bibinfo{number}{6}), p. \bibinfo{pages}{106974},
  \doi{10.1016/j.jpaa.2021.106974}.

\bibitemdeclare{article}{Bacci18}
\bibitem{Bacci18}
\bibinfo{author}{Giorgio \surnamestart Bacci\surnameend},
  \bibinfo{author}{Giovanni \surnamestart Bacci\surnameend},
  \bibinfo{author}{Kim~G. \surnamestart Larsen\surnameend} \&
  \bibinfo{author}{Radu \surnamestart Mardare\surnameend}
  (\bibinfo{year}{2018}): \emph{\bibinfo{title}{A Complete Quantitative
  Deduction System for the Bisimilarity Distance on Markov Chains}}.
\newblock {\slshape \bibinfo{journal}{Logical Methods in Computer Science}}
  \bibinfo{volume}{14}, \doi{10.23638/LMCS-14(4:15)2018}.

\bibitemdeclare{article}{Bacci21}
\bibitem{Bacci21}
\bibinfo{author}{Giorgio \surnamestart Bacci\surnameend},
  \bibinfo{author}{Giovanni \surnamestart Bacci\surnameend},
  \bibinfo{author}{Kim~G. \surnamestart Larsen\surnameend},
  \bibinfo{author}{Radu \surnamestart Mardare\surnameend},
  \bibinfo{author}{Qiyi \surnamestart Tang\surnameend} \&
  \bibinfo{author}{Franck \surnamestart van Breugel\surnameend}
  (\bibinfo{year}{2021}): \emph{\bibinfo{title}{Computing Probabilistic
  Bisimilarity Distances for Probabilistic Automata}}.
\newblock {\slshape \bibinfo{journal}{Logical Methods in Computer Science}}
  \bibinfo{volume}{17}, \doi{10.23638/LMCS-17(1:9)2021}.

\bibitemdeclare{article}{bacci2022sum}
\bibitem{bacci2022sum}
\bibinfo{author}{Giorgio \surnamestart Bacci\surnameend}, \bibinfo{author}{Radu
  \surnamestart Mardare\surnameend}, \bibinfo{author}{Prakash \surnamestart
  Panangaden\surnameend} \& \bibinfo{author}{Gordon~D. \surnamestart
  Plotkin\surnameend} (\bibinfo{year}{2024}): \emph{\bibinfo{title}{Sum and
  Tensor of Quantitative Effects}}.
\newblock {\slshape \bibinfo{journal}{Log. Methods Comput. Sci.}}
  \bibinfo{volume}{20}(\bibinfo{number}{4}), \doi{10.46298/LMCS-20(4:9)2024}.

\bibitemdeclare{inproceedings}{Giry81}
\bibitem{Giry81}
\bibinfo{author}{M.~\surnamestart Giry\surnameend} (\bibinfo{year}{1981}):
  \emph{\bibinfo{title}{A Categorical Approach to Probability Theory}}.
\newblock In \bibinfo{editor}{B.~\surnamestart Banaschewski\surnameend},
  editor: {\slshape \bibinfo{booktitle}{Categorical Aspects of Topology and
  Analysis}}, {\slshape \bibinfo{series}{Lecture Notes In Mathematics}}
  \bibinfo{volume}{915}, \bibinfo{publisher}{Springer-Verlag}, pp.
  \bibinfo{pages}{68--85}.

\bibitemdeclare{article}{Kelly1993}
\bibitem{Kelly1993}
\bibinfo{author}{G.M. \surnamestart Kelly\surnameend} \& \bibinfo{author}{A.J.
  \surnamestart Power\surnameend} (\bibinfo{year}{1993}):
  \emph{\bibinfo{title}{Adjunctions whose counits are coequalizers, and
  presentations of finitary enriched monads}}.
\newblock {\slshape \bibinfo{journal}{Journal of Pure and Applied Algebra}}
  \bibinfo{volume}{89}(\bibinfo{number}{1–2}), p. \bibinfo{pages}{163–179},
  \doi{10.1016/0022-4049(93)90092-8}.

\bibitemdeclare{book}{kelly}
\bibitem{kelly}
\bibinfo{author}{Gregory~Maxwell. \surnamestart Kelly\surnameend}
  (\bibinfo{year}{1982}): \emph{\bibinfo{title}{Basic concepts of enriched
  category theory / Gregory Maxwell Kelly}}.
\newblock \bibinfo{publisher}{Cambridge University Press Cambridge ; New York}.

\bibitemdeclare{article}{Kock1972}
\bibitem{Kock1972}
\bibinfo{author}{Anders \surnamestart Kock\surnameend} (\bibinfo{year}{1972}):
  \emph{\bibinfo{title}{Strong functors and monoidal monads}}.
\newblock {\slshape \bibinfo{journal}{Archiv der Mathematik}}
  \bibinfo{volume}{23}(\bibinfo{number}{1}), p. \bibinfo{pages}{113–120},
  \doi{10.1007/bf01304852}.

\bibitemdeclare{book}{MacLane}
\bibitem{MacLane}
\bibinfo{author}{Saunders~Mac \surnamestart Lane\surnameend}
  (\bibinfo{year}{1978}): \emph{\bibinfo{title}{Categories for the Working
  Mathematician}}.
\newblock \bibinfo{publisher}{Springer New York},
  \doi{10.1007/978-1-4757-4721-8}.

\bibitemdeclare{article}{Lawvere1963}
\bibitem{Lawvere1963}
\bibinfo{author}{F.~W. \surnamestart Lawvere\surnameend}
  (\bibinfo{year}{1963}): \emph{\bibinfo{title}{Functorial Semantics of
  Algebraic Theories}}.
\newblock {\slshape \bibinfo{journal}{Proceedings of the National Academy of
  Sciences}} \bibinfo{volume}{50}(\bibinfo{number}{5}), pp.
  \bibinfo{pages}{869--872}, \doi{10.1073/pnas.50.5.869}.

\bibitemdeclare{article}{Lieberman2017}
\bibitem{Lieberman2017}
\bibinfo{author}{Michael~J. \surnamestart Lieberman\surnameend} \&
  \bibinfo{author}{Jir{\'{\i}} \surnamestart Rosick{\'{y}}\surnameend}
  (\bibinfo{year}{2017}): \emph{\bibinfo{title}{Metric Abstract Elementary
  Classes as Accessible Categories}}.
\newblock {\slshape \bibinfo{journal}{J. Symb. Log.}}
  \bibinfo{volume}{82}(\bibinfo{number}{3}), pp. \bibinfo{pages}{1022--1040},
  \doi{10.1017/JSL.2016.39}.

\bibitemdeclare{inproceedings}{raduQAR2016}
\bibitem{raduQAR2016}
\bibinfo{author}{Radu \surnamestart Mardare\surnameend},
  \bibinfo{author}{Prakash \surnamestart Panangaden\surnameend} \&
  \bibinfo{author}{Gordon~D. \surnamestart Plotkin\surnameend}
  (\bibinfo{year}{2016}): \emph{\bibinfo{title}{Quantitative Algebraic
  Reasoning}}.
\newblock In \bibinfo{editor}{Martin \surnamestart Grohe\surnameend},
  \bibinfo{editor}{Eric \surnamestart Koskinen\surnameend} \&
  \bibinfo{editor}{Natarajan \surnamestart Shankar\surnameend}, editors:
  {\slshape \bibinfo{booktitle}{Proceedings of the 31st Annual {ACM/IEEE}
  Symposium on Logic in Computer Science, {LICS} '16, New York, NY, USA, July
  5-8, 2016}}, \bibinfo{publisher}{{ACM}}, pp. \bibinfo{pages}{700--709},
  \doi{10.1145/2933575.2934518}.

\bibitemdeclare{inproceedings}{raduQAR2017}
\bibitem{raduQAR2017}
\bibinfo{author}{Radu \surnamestart Mardare\surnameend},
  \bibinfo{author}{Prakash \surnamestart Panangaden\surnameend} \&
  \bibinfo{author}{Gordon~D. \surnamestart Plotkin\surnameend}
  (\bibinfo{year}{2017}): \emph{\bibinfo{title}{On the axiomatizability of
  quantitative algebras}}.
\newblock In: {\slshape \bibinfo{booktitle}{32nd Annual {ACM/IEEE} Symposium on
  Logic in Computer Science, {LICS} 2017, Reykjavik, Iceland, June 20-23,
  2017}}, \bibinfo{publisher}{{IEEE} Computer Society}, pp.
  \bibinfo{pages}{1--12}, \doi{10.1109/LICS.2017.8005102}.

\bibitemdeclare{inproceedings}{Mardare21}
\bibitem{Mardare21}
\bibinfo{author}{Radu \surnamestart Mardare\surnameend},
  \bibinfo{author}{Prakash \surnamestart Panangaden\surnameend} \&
  \bibinfo{author}{Gordon~D. \surnamestart Plotkin\surnameend}
  (\bibinfo{year}{2021}): \emph{\bibinfo{title}{Fixed-Points for Quantitative
  Equational Logics}}.
\newblock In: {\slshape \bibinfo{booktitle}{36th Annual {ACM/IEEE} Symposium on
  Logic in Computer Science, {LICS} 2021, Rome, Italy, June 29 - July 2,
  2021}}, \bibinfo{publisher}{{IEEE}}, pp. \bibinfo{pages}{1--13},
  \doi{10.1109/LICS52264.2021.9470662}.

\bibitemdeclare{article}{Moggi1991}
\bibitem{Moggi1991}
\bibinfo{author}{Eugenio \surnamestart Moggi\surnameend}
  (\bibinfo{year}{1991}): \emph{\bibinfo{title}{Notions of Computation and
  Monads}}.
\newblock {\slshape \bibinfo{journal}{Inf. Comput.}}
  \bibinfo{volume}{93}(\bibinfo{number}{1}), pp. \bibinfo{pages}{55--92},
  \doi{10.1016/0890-5401(91)90052-4}.

\bibitemdeclare{article}{Plotkin2003}
\bibitem{Plotkin2003}
\bibinfo{author}{Gordon~D. \surnamestart Plotkin\surnameend} \&
  \bibinfo{author}{John \surnamestart Power\surnameend} (\bibinfo{year}{2003}):
  \emph{\bibinfo{title}{Algebraic Operations and Generic Effects}}.
\newblock {\slshape \bibinfo{journal}{Appl. Categorical Struct.}}
  \bibinfo{volume}{11}(\bibinfo{number}{1}), pp. \bibinfo{pages}{69--94},
  \doi{10.1023/A:1023064908962}.

\bibitemdeclare{article}{power1999}
\bibitem{power1999}
\bibinfo{author}{John \surnamestart Power\surnameend} (\bibinfo{year}{1999}):
  \emph{\bibinfo{title}{Enriched Lavwere Theories}}.
\newblock {\slshape \bibinfo{journal}{Theory and Applications of Categories}}
  \bibinfo{volume}{6}(\bibinfo{number}{7}).

\bibitemdeclare{article}{ros2020metric}
\bibitem{ros2020metric}
\bibinfo{author}{Jir{\'{\i}} \surnamestart Rosick{\'{y}}\surnameend}
  (\bibinfo{year}{2021}): \emph{\bibinfo{title}{Metric monads}}.
\newblock {\slshape \bibinfo{journal}{Math. Struct. Comput. Sci.}}
  \bibinfo{volume}{31}(\bibinfo{number}{5}), pp. \bibinfo{pages}{535--552},
  \doi{10.1017/S0960129521000220}.

\bibitemdeclare{article}{Stone49}
\bibitem{Stone49}
\bibinfo{author}{Marshall~H. \surnamestart Stone\surnameend}
  (\bibinfo{year}{1949}): \emph{\bibinfo{title}{Postulates for the barycentric
  calculus}}.
\newblock {\slshape \bibinfo{journal}{Annali di Matematica Pura ed Applicata}}
  \bibinfo{volume}{29}(\bibinfo{number}{1}), pp. \bibinfo{pages}{25--30}.

\bibitemdeclare{book}{Villani08}
\bibitem{Villani08}
\bibinfo{author}{C{\'e}dric \surnamestart Villani\surnameend}
  (\bibinfo{year}{2008}): \emph{\bibinfo{title}{Optimal transport: old and
  new}}.
\newblock \bibinfo{publisher}{Springer-Verlag}.

\end{thebibliography}


\section*{Appendix}

\begin{proof}[\textbf{Proof of Proposition \ref{prop1}}]
	Clearly $\et{T}$ is generated by a set of equations, so it suffices to check that it's well-founded.
	
	First, we prove that whenever $\Gamma \vdash f((t_i)) \ok \in \eT$ for an operation $f: (N,d)$, we have
	$\Gamma \vdash t_i =_{d(i,j)} t_j$ for each $i, j \in N$.
	To see this, let $\eT' \subseteq \eT$ be the subset with the same equational sequents, but only those well-formedness sequents for which this rule holds.
	Clearly, this contains $S$ (since it contains all equations in $\eT$), so it suffices to prove that it's stable under the inference rules.
	Since it has all the equations in $\eT$, we only have to check $\mathbf{Var}$ and $\mathbf{App}$.
	The former is true by definition, the latter because the preconditions are exactly the extra condition necessary for the postcondition to be in $\eT'$.
	
	Now we prove that $\eT$ is well-founded.
	Again, we define a subset $\eT' \subseteq \eT$. This time, we let it contain the same formational judgments, but only those equational judgments $\Gamma \vdash t =_\epsilon s$ where $\Gamma \vdash t \ok$ and $\Gamma \vdash s \ok$ are in $\eT$.
	Note that this contains all of $S$, so again it suffices to show that it's stable under the inference rules.
	\begin{enumerate}
		\item $\mathbf{Var}$ and $\mathbf{App}$ are now clear.
		\item $\mathbf{Assum}$ holds because $\eT$ satisfies $\mathbf{Var}$.
		\item $\mathbf{Refl}$ holds because the precondition is exactly what's required for the postcondition to be in $\eT'$.
		\item $\mathbf{Symm}$ holds because, given the precondition, we must further have $\Gamma \vdash t \ok \in \eT$ (and also $t'$). Hence, since $\eT$ satisfies $\mathbf{Symm}$, the postcondition must also be in $\eT'$.
		\item $\mathbf{Triang}, \mathbf{Max}, \mathbf{Cont}$ all hold for essentially the same reason as $\mathbf{Symm}$.
		\item To prove $\mathbf{NExp}$, we have to show that, given the preconditions, $\Gamma \vdash f((s_i)) \ok \in \eT$. But this follows from the preconditions, the further fact that because the preconditions are in $\eT'$ we have $\Gamma \vdash s_i \ok \in \eT$ for each $i$, and the fact that $\eT$ satisfies $\mathbf{App}$.
		\item Finally, the most difficult rule is $\mathbf{Subst}$.
		We will prove that, given the preconditions, $\Gamma \vdash s[u_i/x_i] \ok$.
		If $s$ is a variable not among the $x_i$, this is obvious.
		If $s = x_i$, then $s[u_i/x_i] = u_i$, and we are done by assumption.
		Suppose $s = g((s_k))$ for some symbol $g: (N', d')$.
		By induction, assume $\Gamma \vdash s_k[u_i/x_i] \ok \in \eT$ for each $k$.
		By assumption, $\{x_i =_{d(i,j)} x_j\} \vdash g((s_k)) \ok \in \eT$.
		But by the rule we proved above, this means that $\{x_i =_{d(i,j)} x_j\} \vdash s_k =_{d'(k,k')} s_{k'} \in \eT$ for each $k,k' \in N'$.
		Now $\mathbf{Subst}$ for $\eT$ implies that $\Gamma \vdash s_k[u_i/x_i] =_{d'(k,k')} s_{k'}[u_i/x_i]$.
		These claims and the inductive assumption, together with $\mathbf{App}$, imply the desired conclusion.
		\qedhere
	\end{enumerate}
\end{proof}

\begin{proof}[\textbf{Proof of Proposition \ref{prop3}}]
	To give a space $M$ the structure of a model of $\eT(A)$ is to choose for each $a \in A$ a point $M[a] \in M$,
	satisfying $d(M[a],M[a']) \leq d_A(a,a')$. This is equivalent to a map $A \to M$ which is nonexpansive.
	A homomorphism is a nonexpansive map $f: M \to N$ so that $f(M[a]) = N[a]$.
	This is just the definition of a map in the slice category.
\end{proof}

\begin{proof}[\textbf{Proof of Proposition \ref{prop4}}]
	This is essentially true by definition,
	since (by construction) the operations of $\eT$ and $\eT'$ don't overlap within $\eT \sqcup \eT'$,
	and it suffices to satisfy the axioms of each theory independently (since the disjoint union theory is just generated by these).
\end{proof}

\begin{proof}[\textbf{Proof of Lemma \ref{lemma1}}]
	First note that because of $\mathbf{Subst}$, the arbitrary choices don't affect the provability of $\Gamma^A \vdash \phi^A$.
	By using $\mathbf{Subst}$ and $\mathbf{Assum}$, it's easy to see the forwards direction:
	\[\Gamma^A \vdash_{\eT} \phi^A \Rightarrow \Gamma \vdash_{\eT \sqcup \eT(A)} \phi\]
	On the other hand, consider the set of sequents in the signature of $\eT \sqcup \eT(A)$ so that the left-hand side holds.
	Clearly, this set contains both $\eT$ and $\eT(A)$, so it suffices to show it's stable under the inference rules.
	
	\begin{enumerate}
		\item $\mathbf{Var}$ and $\mathbf{Assum}$ are immediate, since $x^A = x$.
		\item To see it's stable under $\mathbf{App}$, suppose first that $f$ is an operation in $\eT$. Note that $\eT$ satisfies $\mathbf{App}$, so given the assumption that $\Gamma^A \vdash_\eT t_i^A \ok$ and so on, we find that $\Gamma^A \vdash_\eT f(t_i^A) \ok$.
		But clearly $f(t_i^A) = f(t_i)^A$, so this is just what we wanted.
		On the other hand, if $f = [a]$ is a nullary symbol, this is automatically true.
		
		An analogous argument proves $\mathbf{Refl}, \mathbf{Symm}, \mathbf{Triang}, \mathbf{Cont}$, and $\mathbf{NExp}$.
		\item For $\mathbf{Subst}$, we must be a bit more careful.
		
		Let $x_i, \delta_{ij}, u_i, s, t, \Gamma$ be as in the assumptions of the inference rule, where the terms are terms over $\eT \sqcup \eT(A)$.
		We assume we've further chosen our $x_a$ among variables not occurring in these terms.
		We are assuming that $\Gamma^A \vdash u_i^A =_{\delta_{ij}} u_j^A$ and $\{x_i =_{\delta_{ij}} x_j \mid i,j\}^A \vdash s^A =_\epsilon t^A$.
		
		Now, $\{x_i =_{\delta_{ij}} x_j\}^A = \{x_i =_{\delta_{ij}} x_j\} \cup \{x_a =_{d(a,a')} x_{a'}\}$.
		Note that $\Gamma^A \vdash x_a =_{d(a,a')} x_{a'}$ for all $a,a'$.
		So by expanding to the set of variables containing both the $x_i$ and the $x_a$, and taking $u_a = x_a$, we can apply $\mathbf{Subst}$ for $\eT$ to prove that $\Gamma^A \vdash s^A[u_i^A/x_i] =_\epsilon t^A[u_i^A/x_i]$.
		Now we just observe that $s^A[u_i^A/x_i] = s[u_i/x_i]^A$ and we're done.
		\qedhere
	\end{enumerate}
\end{proof}

\begin{proof}[\textbf{Proof of Proposition \ref{prop5}}]
		Let $X$ be a complete metric space.
		Given a sequence $(x_1, \dots x_n) \in X^N$, note that $d(x_n,x_m) \leq 1/2^n$ if $m> n$, so this is a Cauchy sequence.
		By continuity of the metric, $d(x_n,\lim_i x_i) = \lim_i d(x_n, x_i)$, which is less than  or equal $1/2^n$ from a certain point, so $d(x_n, \lim_i x_i) \leq 1/2^n$.
		Thus defining $X[\lim](x_1, \dots) = \lim_i x_i$ makes $X$ a model.
		On the other hand, any number $l$ satisfying $d(l,x_n) \leq 1/2^n$ for all $n$ must clearly be the limit, so this is the only way to make $X$ a model. Since nonexpansive maps are continuous, any nonexpansive map $X \to Y$ between complete metric spaces is a homomorphism between their associated models.
		
		On the other hand, let $M$ be a model of $\eT^\mathrm{comp}$.
		Let $a_1, \dots a_n$ be a Cauchy sequence.
		Then we can find $n_1$ so that $d(a_{n_1}, a_m) \leq 1/2$ for all $m> n_1$, $n_2 > n_1$ so that $d(a_{n_2}, a_m) \leq 1/2^2$ for $m> n_2$, and so on, since it is a Cauchy sequence.
		Now the subsequence $a_{n_1}, a_{n_2}, \dots$ is an element of $M^N$.
		Clearly $M[\lim](a_{n_1}, \dots)$ is a limit of this subsequence, hence a limit of the original sequence (since it is Cauchy).
		Hence, any model of $\eT^\mathrm{comp}$ is complete.
		By the preceding, it $M[\lim]$ must be given by taking limits, and any nonexpansive map between models must be a homomorphism.
		This concludes the proof.    
	\end{proof}

\end{document}